\documentclass{tlp}
\usepackage{amssymb}
\usepackage{amsmath}
\usepackage{url}
\usepackage{stmaryrd}
\usepackage{color}

\usepackage[hidelinks]{hyperref} 
\usepackage[user,xr]{zref}

\zexternaldocument{appendix_paper42}

\newcommand{\eqdef}{\mathbin{\stackrel{\mathrm{def}}{=}}}
\def\cI{{\cal I}}
\def\cIc{{\cal I}_{\, c}}
\def\down{\downarrow}
\def\up{\uparrow}
\def\L{{\cal L}}
\newcommand\non[1]{\overline{#1} }
\newcommand{\den}[1]{\llbracket \, #1 \, \rrbracket}

\def\Not{\text{not} \ }
\def\qed{~\hfill$\Box$}
\newtheorem{definition}{Definition}
\newtheorem{theorem}{Theorem}
\newtheorem{example}{Example}
\newtheorem{lemma}{Lemma}

\newtheorem{corollary}{Corollary}
\newtheorem{proposition}{Proposition}

\newcommand{\str}[1]{\underline{#1}}

\def\eqmodels{\ |\!\!\!\sim}

\title[Theory and Practice of Logic Programming]
      {A Denotational Semantics \\ for Equilibrium Logic \thanks{This research was partially supported by Spanish MEC project TIN2013-42149-P.}}

\author[F. Aguado et al]
         {Felicidad Aguado$^1$, Pedro Cabalar$^1$, David Pearce$^2$, \\ 
         {\rm \normalsize Gilberto P\'erez$^1$ and Concepci\'on Vidal$^1$}\\ \\
         $^1$ Department of Computer Science\\
         University of Corunna, SPAIN\\
         \email{\{aguado,cabalar,gperez,eicovima\}@udc.es}\\
\\
         $^2$ Universidad Polit\'ecnica de Madrid, SPAIN\\
         \email{david.pearce@upm.es}
		 }

\begin{document}

\label{firstpage}

\maketitle

\begin{abstract}
In this paper we provide an alternative semantics for Equilibrium Logic and its monotonic basis, the logic of Here-and-There (also known as G\"odel's $G_3$ logic) that relies on the idea of \emph{denotation} of a formula, that is, a function that collects the set of models of that formula. Using the three-valued logic $G_3$ as a starting point and an ordering relation (for which equilibrium/stable models are minimal elements) we provide several elementary operations for sets of interpretations. By analysing structural properties of the denotation of formulas, we show some expressiveness results for $G_3$ such as, for instance, that conjunction is not expressible in terms of the other connectives. Moreover, the denotational semantics allows us to capture the set of equilibrium models of a formula with a simple and compact set expression. We also use this semantics to provide several formal definitions for entailment relations that are usual in the literature, and further introduce a new one called \emph{strong entailment}. We say that $\alpha$ strongly entails $\beta$ when the equilibrium models of $\alpha \wedge \gamma$ are also equilibrium models of $\beta \wedge \gamma$ for any context $\gamma$. We also provide a characterisation of strong entailment in terms of the denotational semantics, and give an example of a sufficient condition that can be applied in some cases.
\end{abstract}


\section{Introduction}

In the last 15 years, the paradigm of Answer Set Programming (ASP)~\cite{MT99,Nie99,BET11} has experienced a boost in practical tools and applications that has come in parallel with a series of significant results in its theoretical foundations. Focusing on the latter, a long way has been traversed since the original definition of the \emph{stable models} semantics~\cite{GL88} for normal logic programs, until the current situation where stable models constitute a complete non-monotonic approach for arbitrary theories in the syntax of First Order Logic~\cite{PV04,FLL07}. An important breakthrough that undoubtfully contributed to this evolution was the characterization of stable models in terms of \emph{Equilibrium Logic}~\cite{Pea96,Pea06}, allowing a full coverage of arbitrary propositional theories and inspiring a new definition of program reduct for that syntax~\cite{Fer05}. Equilibrium Logic is defined in terms of a model minimisation criterion for an intermediate logic called the \emph{Logic of Here-and-There} (HT) first introduced in~\cite{Hey30} and, shortly after, reappeared in~\cite{God32} as G\"odel's three-valued logic $G_3$. In~\cite{LPV01} it was shown that equivalence in HT was a necessary and sufficient condition for the property of \emph{strong equivalence}, that is, that two programs yield the same stable/equilibrium models regardless of the context in which they may be included. After that, many theoretical results have followed from the use of Equilibrium Logic and HT, such as the study of variants of strong equivalence~\cite{PV04b,Wol08} or the series of papers considering different forms of strongly equivalent transformations~\cite{CPV05,CF07,CPV07}. Besides, Equilibrium Logic allowed the already mentioned extension to first order syntax~\cite{PV04}, engendering an extensive literature, as well as many other extensions such as the inclusion of a strong negation operator~\cite{OP05} or new formalisms such as \emph{Partial Equilibrium Logic}~\cite{COPV07b}, \emph{Temporal Equilibrium Logic}~\cite{ACD+13} or, more recently, \emph{Infinitary Equilibrium Logic}~\cite{HLPV14}.

All these contributions provide results about HT or Equilibrium Logic that are proved with meta-logical textual descriptions. These proofs lack a common formal basis on which meta-properties of HT and Equilibrium Logic can be mathematically or even automatically checked. Another interesting observation is that many of these theoretical results in the literature use the concept of \emph{sets of models} of different types: classical models, HT models, equilibrium models, etc. It is, therefore, natural to wonder whether a formal treatment of sets of interpretations could help in the development of fundamental results for Equilibrium Logic and ASP.

In this paper we explore the idea of characterising HT (or $G_3$) and Equilibrium Logic using the concept of \emph{denotation} of a formula. Given a formula $\alpha$, its denotation $\den\alpha$ collects the set of $G_3$ models of $\alpha$ and can be described as a compositional function, that is, the denotation of a formula is a function of the denotations of its subformulas. Since their introduction by~\cite{SS71}, denotations constitute a common device for defining the semantics of programming languages, although their use for non-classical logics is also frequent -- a prominent case, for instance, is the semantics of $\mu$-Calculus~\cite{Koz83}. The use of denotational semantics in Logic Programming is not so common: in the case of Prolog we can mention~\cite{NF89} but for ASP, to the best of our knowledge, no attempt has previously been made.

Here we explain how the denotational semantics actually constitutes an alternative description of HT/$G_3$ and provides several interesting features. We define some elementary operations on sets of interpretations and the ordering relation used in the equilibrium models minimisation. Using those elementary set operations and analysing structural properties of the denotation of formulas, we derive some expressivity results for $G_3$ such as, for instance, that conjunction is not expressible in terms of implication, falsum and disjunction. More importantly, we are able to capture the equilibrium models of a formula as a set expression constituting a subset of $\den\alpha$. This allows us to study properties of equilibrium models by using formal results from set theory, something that in many cases is more compact than an informal proof in natural language and, moreover, has allowed us to use a theorem prover for a semi-automated verification (see the sequel~\cite{CMMSE15} of the current paper).

As an application of the denotational semantics, we provide several definitions (in terms of denotations) for entailment relations foundl in the literature, and further introduce a new one called \emph{strong entailment}. We say that $\alpha$ strongly entails $\beta$ when the equilibrium models of $\alpha \wedge \gamma$ are also equilibrium models of $\beta \wedge \gamma$ for any context $\gamma$. This obviously captures one of the directions of strong equivalence. We also provide the corresponding denotational characterisation for this new strong entailment and give an example of a sufficient condition that can be applied in some cases.

The rest of the paper is organised as follows. In Section~\ref{sec:g3} we provide the basic definitions of G\"odel's $G_3$ logic that, as explained, is an equivalent formulation of HT. In Section~\ref{sec:sets} we describe several useful operators on sets of interpretations that we then use in Section~\ref{sec:zero} to define the denotational semantics for $G_3$ and for equilibrium models. After describing some applications of this semantics, Section~\ref{sec:entail} defines different types of entailments and, in particular, presents the idea of strong entailment together with its denotational charaterisation and some examples. Finally, Section~\ref{sec:conc} concludes the paper. Most proofs have been collected in the on line appendix (Appendix A).

\section{G\"odel's three-valued logic $G_3$ and equilibrium models}
\label{sec:g3}
We describe next the characterisation of Equilibrium Logic in terms of G\"odel's three-valued logic -- for further details on multi-valued characterisations of Equilibrium Logic, see~\cite{Pea06}, section 2.4. \\
We start from a finite set of atoms $\Sigma$ called the \emph{propositional signature}. A \emph{formula} $\alpha$ is defined by the grammar: 
\[
\alpha ::= \bot \mid  p \mid \alpha_1 \wedge \alpha_2 \mid \alpha_1 \vee \alpha_2 \mid \alpha_1 \rightarrow \alpha_2
\]

\noindent where $\alpha_1$ and $\alpha_2$ are formulas in their turn and $p \in \Sigma$ is any atom. We define the derived operators $\neg \alpha \eqdef \alpha \rightarrow \bot$ and $\top \eqdef \neg \bot$. By $\L_{\Sigma}$ we denote the language of all well-formed formulas for signature $\Sigma$ and just write $\L$ when the signature is clear from the context.

A \emph{partial} (or \emph{three-valued}) interpretation is a mapping $v: \Sigma \rightarrow \{0,1,2\}$ assigning $0$ (false), $2$ (true) or $1$ (undefined) to each atom $p$ in the signature $\Sigma$. A partial interpretation $v$ is said to be \emph{classical} (or \emph{total}) if $v(p)\neq 1$ for every atom $p$. We write $\cI$ and $\cIc$ to stand for the set of all partial and total interpretations, respectively (fixing signature $\Sigma$). Note that $\cIc \subseteq \cI$.

For brevity, we will sometimes represent interpretations by (underlined) strings of digits from $\{0,1,2\}$ corresponding to the atom values, assuming the alphabetical ordering in the signature. Thus, for instance, if $\Sigma=\{p,q,r\}$, the interpretation $v=\str{102}$ stands for $v(p)=1$, $v(q)=0$ and $v(r)=2$.

Given any partial interpretation $v \in \cI$ we define a classical interpretation $v_t \in \cIc$ as:
\begin{eqnarray*}
v_t(p) & \eqdef & \left\{
\begin{array}{r@{\ \ }l}
2 & \hbox{if } v(p)=1 \\
v(p) & \hbox{otherwise}
\end{array}
\right.
\end{eqnarray*}
\noindent In other words, $v_t$ is the result of transforming all $1$'s in $v$ into $2$'s. For instance, given $v'=\str{1021}$ for signature $\Sigma=\{p,q,r,s\}$, then $v'_t=\str{2022}$.

\begin{definition}[Valuation of formulas]\label{def:val}
Given a partial interpretation $v \in \cI$ we define a corresponding \emph{valuation of formulas}, a function also named $v$ (by abuse of notation) of type $v: \L \rightarrow \{0,1,2\}$ and defined as:
\[
\begin{array}{c@{\hspace{30pt}}c}
\begin{array}{rcl}
v(\alpha \wedge \beta) & \eqdef & \min(v(\alpha), v(\beta))\\[5pt]
v(\alpha \vee \beta) & \eqdef & \max(v(\alpha), v(\beta)) 
\end{array}
&
\begin{array}{rcl}
v(\bot) & \eqdef & 0 \\[5pt]
v(\alpha \rightarrow \beta) & \eqdef & \left\{
\begin{array}{r@{\ \ }l}
2 & \hbox{if } v(\alpha) \leq v(\beta) \\
v(\beta) & \hbox{otherwise} \hspace{14pt} \Box
\end{array}
\right.
\end{array}
\end{array}
\]
\end{definition}


From the definition of negation, it is easy to see that $v(\neg \alpha)=2$ iff $v(\alpha)=0$, and $v(\neg \alpha)=0$ otherwise.
We say that $v$ \emph{satisfies} $\alpha$ when $v(\alpha)=2$. We say that $v$ is a \emph{model} of a theory $\Gamma$ iff $v$ satisfies all the formulas in $\Gamma$.

\begin{example}\label{ex:1}
As an example, looking at the table for implication, the models of the formula:
\begin{eqnarray}
\neg p \rightarrow q \label{f1}
\end{eqnarray}
\noindent are those where $v(\neg p)=0$ or $v(q)=2$ or both $v(\neg p)=v(q)=1$. The latter is impossible since the evaluation of negation never returns $1$, whereas $v(\neg p)=0$ means $v(p)\neq 0$. Therefore, we get $v(p) \neq 0$ or $v(q)=2$ leading to the following 7 models $\str{10}, \str{11}, \str{12}, \str{20}, \str{21}, \str{22}, \str{02}$.  \qed
\end{example}

Given two 3-valued interpretations $u,v$, we say that $u \leq v$ when, for any atom $p \in \Sigma$, the following two conditions hold: $u(p) \leq v(p)$; and $u(p)=0$ implies $v(p)=0$. As usual, we write $u < v$ to stand for both $u\leq v$ and $u\neq v$. An equivalent, and perhaps simpler, way of understanding $u \leq v$ is that we can get $v$ by switching some $1$'s in $u$ into $2$'s. This immediately means that classical interpretations are $\leq$-maximal, because they contain no $1$'s. Moreover, since $u_t$ is the result of switching \emph{all} $1$'s in $u$ into $2$'s, we easily conclude $u \leq u_t$ for any $u$. As an example of how $\leq$ works, among models of \eqref{f1}, we can check that $\str{10} < \str{20}$ and that $\str{11}$, $\str{12}$ and $\str{21}$ are strictly smaller than $\str{22}$. On the other hand, for instance, $\str{10}$, $\str{02}$ or $\str{12}$ are all pairwise incomparable.

Once we have defined an ordering relation among interpretations, we can define the concept of \emph{equilibrium model} as a $\leq$-minimal model that is also classical.

\begin{definition}[Equilibrium model]
A classical interpretation $v \in \cIc$ is an \emph{equilibrium model} of a theory $\Gamma$ iff it is a $\leq$-minimal model of $\Gamma$.\qed
\end{definition}

Back to the example \eqref{f1}, from the 7 models we obtained, only three of them $\str{20}, \str{22}$ and $\str{02}$ are classical (they do not contain $1$'s). However, as we saw, $\str{20}$ is not $\leq$-minimal since $\str{10}<\str{20}$ and the same happens with $\str{22}$, since $\str{11}, \str{12}, \str{21}$ are strictly smaller too. The only $\leq$-minimal classical model is $\str{02}$, that is, $p$ false and $q$ true, which becomes the unique equilibrium model of \eqref{f1}. Equilibrium models coincide with the most general definition of stable models, for the syntax of arbitrary (propositional) formulas~\cite{Fer05}. Indeed, we can check that model $\str{02}$ coincides with the only stable model of the ASP rule $(q \leftarrow \Not p)$ which is the usual rewriting of formula \eqref{f1} in ASP syntax. 

\section{Sets of interpretations}\label{sec:sets}

In this section we will introduce some useful operations on sets of interpretations. Some of them depend on the partial ordering relation $\leq$. Given a set of interpretations $S\subseteq \cI$ we will define the operations:
\[
\begin{array}{rcl@{\hspace{50pt}}rcl}
\non{S} & \eqdef & \cI \setminus S &
S \down & \eqdef & \{u \in \cI \ : \ \hbox{there exists } v \in S, v \geq u\} \\
S_c & \eqdef & \cIc \cap S &
S \up & \eqdef & \{u \in \cI \ : \ \hbox{there exists } v \in S, v \leq u\} 
\end{array}
\]
To avoid too many parentheses, we will assume that $\down$, $\up$ and $c$ have more priority than standard set operations $\cup$, $\cap$ and $\setminus$. As usual, we can also express set difference $S \setminus S'$ as $S \cap \non{S'}$. We can easily check that the $c$ operation distributes over $\cap$ and $\cup$, whereas $\down$ and $\up$ distribute over $\cup$. For intersection, we can only prove that:

\begin{proposition}\label{prop:distr} \zlabel{prop:distr}
For any pair $S, S'$ of sets of interpretations: \\
$(S \cap S')\up \ \ \subseteq \ S\up \cap \ S'\up$ and $(S \cap S')\down \ \ \subseteq \ S\down \cap \ S'\down$.\qed
\end{proposition}

\noindent In the general case, the other direction does not hold. As a simple example, for signature $\Sigma=\{p,q\}$, take $S=\{\str{12}\}$ and $S'=\{\str{21}\}$. Then $S\up = \{\str{12},\str{22}\}$ and $S'\up =\{\str{21},\str{22}\}$ and thus $S\up \cap \ S'\up=\{\str{22}\}$ but $(S\cap S')\up = \emptyset$.

With these new operators we can formally express that $v_t$ is the only classical interpretation greater than or equal to $v$ in the following way:

\begin{proposition}\label{prop:vt} \zlabel{prop:vt}
For any $v \in \cI$ it holds that $\{v\} \up_c \ = \{ v_t \}$. \qed
\end{proposition}
\begin{corollary}\label{cor:vt} \zlabel{cor:vt}
For any $S \subseteq \cI$ and for any interpretation $v$ we have: $v \in S_c \down$ iff $v_t \in S$.\qed
\end{corollary}

A particularly interesting type of sets of interpretations are those $S$ satisfying that, for any $v \in S$, we also have $v_t \in S$. When this happens, we say that $S$ is \emph{total-closed} or \emph{classically closed}. As we will see, there is a one-to-one correspondence between a total-closed set of interpretations and a set of models for some (set of equivalent) formula(s). The definition of total-closed set can be formally captured as follows:

\begin{proposition}\label{prop:totalc} \zlabel{prop:totalc}
The following three assertions are equivalent:\\
(i) \ $S$ is total-closed
\hspace{45pt}
(ii) \ $S\subseteq S_c \down$
\hspace{45pt}
(iii) \ $S \up_c = S_c$.\qed
\end{proposition}

\begin{lemma}\label{lem:cm} \zlabel{lem:cm}
For any set of interpretations $S$, it holds that $(\non{S})_c \down \subseteq \non{(S_c \down)}$.\qed
\end{lemma}
From this, together with Proposition \ref{prop:totalc} (ii) we immediately conclude
\begin{proposition}\label{prop:cm} \zlabel{prop:cm}
For any total-closed set of interpretations $S$, it holds that $(\non{S})_c \down \subseteq \non{S}$.\qed
\end{proposition}
\noindent When $S$ is a total-closed set of models, this proposition asserts that any interpretation below a classical countermodel is also a countermodel. In fact, Proposition~\ref{prop:cm} corresponds to what~\cite{CF07} defined as \emph{total-closed set of countermodels} $\non{S}$.


\section{Denotational semantics}
\label{sec:zero}

In this section we consider a denotational semantics for $G_3$ and for equilibrium models. Rather than saying when an interpretation $v$ is a model of a formula $\varphi$, the main idea is to capture the \emph{whole set of models} of $\varphi$ as a set of interpretations we will denote by $\den{\varphi}$. As we explain next, this set can be completely defined by structural induction without actually resorting to the valuation of formulas.

\begin{definition}[Denotation]
The \emph{denotation} of a formula $\varphi$, written $\den{\varphi}$, is recursively defined as follows
\[
\begin{array}{rcl@{\hspace{30pt}}rcl}
\den{\bot} & \eqdef & \emptyset &
\den{\alpha \wedge \beta} & \eqdef & \den{\alpha} \cap \den{\beta}\\
\den{p} & \eqdef & \{v \in \cI \; : \; v(p)=2 \} &
\den{\alpha \vee \beta} & \eqdef & \den{\alpha} \cup \den{\beta} \\
\den{ \alpha \rightarrow \beta} & \eqdef & \big( \non{\den{\alpha}} \cup \den{\beta} \big) \cap \big( \ \non{\den{\alpha}} \cup \den{\beta} \ \big)_c\down 
\end{array}
\]
\noindent where $p \in \Sigma$ is an atom, and $\alpha, \beta \in \L$ are formulas in their turn.\qed 
\end{definition}

We say that a formula $\alpha$ is a \emph{tautology} iff $\den\alpha=\cI$ and that the formula is \emph{inconsistent} iff $\den\alpha=\emptyset$. The following theorem shows that this definition actually captures the set of models of $\alpha$, i.e., the set of interpretations that make $v(\alpha)=2$ using $G_3$ valuations of formulas (Definition~\ref{def:val}). Moreover, it also proves that $v_t \in \den{\alpha}$ is equivalent to $v(\alpha)\neq 0$.

\begin{theorem}\label{th:den-g3} \zlabel{th:den-g3}
Let $v \in \cI$ be a partial interpretation and $\alpha \in \L$ a formula. Then:
\begin{enumerate}
\item[\rm (i)] $v(\alpha)=2$ in $G_3$ iff $v \in \den\alpha$.
\item[\rm (ii)] $v(\alpha) \neq 0$ in $G_3$ iff $v_t \in \den\alpha$.\qed
\end{enumerate}
\end{theorem}

\noindent As $v(\alpha)=2$ implies $v(\alpha)\neq 0$, then $v\in \den\alpha$ implies $v_t \in \den\alpha$ and thus: 
\begin{corollary}\label{cor:vt1}
For any $\alpha \in \L$, $\den\alpha$ is total-closed.\qed
\end{corollary}

\noindent In fact, this relation between models of a formula and total-closed sets of interpretations also holds in the other direction, that is, for any total-closed set of interpretations $S$, there always exists\footnote{This was proved in Theorem 2 from~\cite{CF07} using the dual concept of total-closed set of countermodels.} a formula $\alpha$ such that $\den\alpha = S$.

When compared to denotational semantics for other formalisms, it is clear that the denotation of implication is the most representative characteristic of $G_3$. Defining its denotation provides a powerful tool for studying fundamental properties of this logic. For instance, we can derive the denotation for negation as $\den{\neg \alpha}=\den{\alpha \rightarrow \bot}= \non{\den{\alpha}} \cap \non{\den{\alpha}}_c \down = \non{\den{\alpha}}_c \down$ where the last step follows from Proposition~\ref{prop:cm}. With this correspondence and Corollary~\ref{cor:vt} we conclude that $v \in \den{\neg \alpha} $ iff $v_t \in \non{\den{\alpha}}$, that is, $v$ is a model of $\neg \alpha$ iff $v_t$ is a classical countermodel of $\alpha$. Another application example of the denotation of implication is, for instance, this simple proof of the Deduction Theorem for $G_3$.

\begin{theorem}
For any pair of formulas $\alpha, \beta$: $\den\alpha \subseteq \den\beta$ iff $\den{\alpha \rightarrow \beta} =\cI$. Moreover, $\den\alpha=\den\beta$ iff $\den{\alpha \leftrightarrow \beta} =\cI$.
\end{theorem}

\begin{proof}
For the result with implication, from left to right, assume $\den\alpha \subseteq \den\beta$. Then, $\big( \non{\den{\alpha}} \cup \den{\beta} \big) = \cI$ and so, $\den{\alpha \rightarrow \beta}=\cI \cap \cI_c\down=\cI$. For right to left, if $\den{\alpha \rightarrow \beta} =\cI$, take any $v \in \den\alpha$. As $v \in \den{\alpha \rightarrow \beta} \subseteq \big( \non{\den{\alpha}} \cup \den{\beta} \big)$ we conclude $v \in \den\beta$. For the double implication, simply note that $\den\alpha=\den\beta$ now means $\den{\alpha \rightarrow \beta}=\den{\beta \rightarrow \alpha}=\cI$. Therefore, $\den{\alpha \leftrightarrow \beta}=\den{\alpha \rightarrow \beta} \cap \den{\beta \rightarrow \alpha}=\cI\cap\cI=\cI$.
\end{proof}

\noindent This denotation of implication is an intersection of two sets. We can also alternatively capture implication as a union of sets:
\begin{proposition}\label{prop:vt2} \zlabel{prop:vt2}
For any $\alpha, \beta \in \L$, it follows that:
\begin{eqnarray*}
\hspace{68pt} \den{ \alpha \rightarrow \beta} & = & \overline{\den\alpha}_c \down \ \cup \ ( \overline{\den{\alpha}} \cap \den\beta_c \down) \ \cup \ \den\beta
\hspace{68pt}\Box 
\end{eqnarray*}
\end{proposition}

From this alternative representation of implication and the fact that $\non{\den{\alpha}}_c\down \ \subseteq \non{\den\alpha}$ (from Proposition~\ref{prop:cm}) we immediately conclude $\den\alpha \cap \den{\alpha \rightarrow \beta} = \den\alpha \cap \den\beta$. In other words, we have trivially proved that $\den{\alpha \wedge (\alpha \rightarrow \beta)}=\den{\alpha \wedge \beta}$ in $G_3$.

\subsection{Expressiveness of operators}

As an application of the denotational semantics, we will study the expressiveness of the set of propositional operators usually provided as a basis for $G_3$: $\{\wedge,\vee,\rightarrow,\bot\}$. In Intuitionistic Logic, it is well-known that we cannot represent any of these operators in terms of the others. In $G_3$, however, it is also known 
that $\vee$ can be represented in terms of $\wedge$ and $\rightarrow$. In particular:

\begin{theorem}\label{th:or} \zlabel{th:or}
For any $\Sigma$, the system $\L_{\Sigma}\{\bot, \wedge, \rightarrow\}$ is complete because given any pair of formulas $\alpha, \beta$ for $\Sigma$, it holds that:
$\den{ \alpha \vee \beta} = \den{(\alpha \rightarrow  \beta) \rightarrow \beta} \cap \den{(\beta \rightarrow  \alpha) \rightarrow \alpha}$. \qed
\end{theorem}

Now, one may wonder whether $\rightarrow$ or $\wedge$ can be expressed in terms of the rest of operators. However, we prove next that this is not the case.

\begin{lemma}\label{lem:imply} \zlabel{lem:imply}
Let $\Sigma=\{p_1,\dots,p_n\}$ and let $\gamma \in \L_{\Sigma}\{\bot, \wedge, \vee\}$. Then $\den\gamma \subseteq \bigcup^n_{i=1} \den{p_i}$. \qed
\end{lemma}

\begin{theorem}\label{th:imp} \zlabel{th:imp}
If $\{p_1,p_2\}\subseteq \Sigma$ then $p_1 \rightarrow p_2$ cannot be equivalently represented in $\L_{\Sigma}\{\bot, \vee, \wedge\}$.\qed
\end{theorem}

\noindent This result is not surprising since we can further observe that the denotations for $\wedge$ and $\vee$, respectively the intersection and the union, are monotonic with respect to set inclusion, whereas $\den{\alpha \rightarrow \beta}$ is monotonic for the consequent and anti-monotonic for the antecedent (see Proposition~\zref{prop:mon} in the online appendix).

We will show next that conjunction cannot be expressed in terms of $\vee, \rightarrow, \bot$. To this aim, we begin proving the following lemma.

\begin{lemma}\label{lem:pq} \zlabel{lem:pq}
Let $\Sigma$ be of the form $\Sigma=\{p,q,\dots\}$ and let $\gamma \in \L_{\Sigma}\{\bot, \vee, \rightarrow\}$, then for any subformula $\delta$ of $\gamma$ and any $v \in \den\delta$ of the form $v=\str{22\dots}$ (i.e. making both atoms true), there exists some $u \in \den{\delta}$ such that $u<v$ and $u$ coincides with $v$ in all atoms excepting $p,q$.\qed
\end{lemma}

\begin{theorem}\label{th:wedge} \zlabel{th:wedge}
If $\{p_1,p_2\}\subseteq \Sigma$ then $p_1 \wedge p_2$ cannot be equivalently represented in $\L_{\Sigma}\{\bot, \vee, \rightarrow\}$.\qed
\end{theorem}

\subsection{Denotation of equilibrium models}

We can use the denotational semantics to capture equilibrium models as follows.
\begin{theorem}\label{th:1eq} \zlabel{th:1eq}
A classical interpretation $v \in \cIc$ is an equilibrium model of $\alpha$ iff it  satisfies the fixpoint condition $\den{\alpha} \cap \{v\}\down \ = \{v\}$.\qed
\end{theorem}

The set of equilibrium models can also be captured as the denotation below.
\begin{theorem}\label{th:eq} \zlabel{th:eq}
The set of equilibrium models of $\alpha$, denoted as $\den{\alpha}_e$, corresponds to the expression:
\[
\hspace{110pt}
 \den{\alpha}_e \eqdef \den{\alpha}_c \ \setminus \ (\den{\alpha} \setminus \cIc ) \up \hspace{110pt} \Box
\]
\end{theorem}

As an application of Theorem~\ref{th:eq}, we have used it to obtain the following characterisation of equilibrium models of a disjunction:

\begin{proposition}\label{prop:or}
For any pair of formulas $\alpha$ and $\beta$:
\begin{eqnarray*}
\den{\alpha \vee \beta}_e & = & 
\big( \den{\alpha}_e \setminus \den{\beta}_c \big)
\ \cup \
\big( \den{\beta}_e \setminus \den{\alpha}_c \big)
\ \cup \
\big( \den{\alpha}_e \cap \den{\beta}_e \big)
\end{eqnarray*}
\end{proposition}
\begin{proof}
We begin applying some basic set operations:
\begin{eqnarray*}
\den{\alpha \vee \beta}_e & = & \den{\alpha \vee \beta}_c \ \setminus \ (\den{\alpha \vee \beta} \setminus \cIc ) \up  \\
& = &\big( \den{\alpha}_c \cup \den{\beta}_c \big)\ \setminus \ (\ (\den{\alpha}\setminus \cIc) \cup (\den{\beta} \setminus \cIc) \ ) \up  \\
& = &\big( \den{\alpha}_c \cup \den{\beta}_c \big)\ \setminus \ ( \ (\den{\alpha}\setminus \cIc)\up \cup \ (\den{\beta} \setminus \cIc)\up \ )  \\
& = &\big( \den{\alpha}_c \cup \den{\beta}_c \big)\ \cap \ \non{(\den{\alpha}\setminus \cIc)\up} \cap \ \non{(\den{\beta} \setminus \cIc)\up}  \\
& = &\den{\alpha}_c  \cap \ \non{(\den{\alpha}\setminus \cIc)\up} \cap \ \non{(\den{\beta} \setminus \cIc)\up}  \\
& & \cup \ \den{\beta}_c \cap \ \non{(\den{\alpha}\setminus \cIc)\up} \cap \ \non{(\den{\beta} \setminus \cIc)\up}  \\
& = &\underbrace{\den{\alpha}_e  \cap \ \non{(\den{\beta} \setminus \cIc)\up}}_{\gamma_1}  \ \ \cup \ \ \underbrace{\den{\beta}_e \cap \ \non{(\den{\alpha}\setminus \cIc)\up}}_{\gamma_2}
\end{eqnarray*}
Since $\den{\alpha}_e \subseteq \cI_c = (\den{\beta}\cup \non{\den{\beta}})_c = \den{\beta}_c \cup \non{\den\beta}_c$ we can rewrite $\gamma_1$ as follows:
\begin{eqnarray*}
\gamma_1 & = & \den{\alpha}_e \cap \big( \den{\beta}_c \cup \non{\den\beta}_c  \big) \cap \ \non{(\den{\beta} \setminus \cIc)\up} \nonumber \\
& = & \den{\alpha}_e \cap \underbrace{\den{\beta}_c \cap \ \non{(\den{\beta} \setminus \cIc)\up}}_{\den{\beta}_e} \ \cup \ \den{\alpha}_e \cap \non{\den{\beta}}_c \cap \ \non{(\den{\beta} \setminus \cIc)\up} \nonumber \\
& = & \den{\alpha}_e \cap \den{\beta}_e \ \cup \ \den{\alpha}_e \cap \non{\den{\beta}}_c \cap \ \non{(\den{\beta} \setminus \cIc)\up}
 \end{eqnarray*}
Now, we will prove that $\non{\den{\beta}}_c \subseteq \non{(\den{\beta} \setminus \cIc) \up}$ and so, we can remove the latter in $\gamma_1$. To this aim, we will show that $\non{\den{\beta}}_c \cap (\den{\beta} \setminus \cIc) \up =\emptyset$. First, note that $\non{\den{\beta}}_c \cap (\den{\beta} \setminus \cIc) \up = \non{\den{\beta}}_c \cap (\den{\beta} \setminus \cIc) \up_c$. Then $(\den{\beta} \setminus \cIc) \up_c = (\den{\beta} \cap \non{\cIc}) \up_c \subseteq \den{\beta}\up_c \cap\ \non{\cIc} \up_c \subseteq \den{\beta}\up_c = \den{\beta}_c$ where, in the last step, we have used Proposition~\ref{prop:totalc} (iii). Finally, as $\non{\den{\beta}}_c \cap \den{\beta}_c = \emptyset$, we conclude $\non{\den{\beta}}_c \cap (\den{\beta} \setminus \cIc) \up_c = \non{\den{\beta}}_c \cap (\den{\beta} \setminus \cIc) \up= \emptyset$ too.

Therefore, we can further simplify the expression we obtained for $\gamma_1$ as:
\begin{eqnarray*}
\gamma_1 & = & \den{\alpha}_e \cap \den{\beta}_e \ \cup \ \den{\alpha}_e \cap \non{\den{\beta}}_c \ = \ 
 \den{\alpha}_e \cap \den{\beta}_e \ \cup \ \big(\den{\alpha}_e \setminus \den{\beta}_c\big)
\end{eqnarray*}
Finally, making a similar reasoning for $\gamma_2$ we get $\gamma_2= \den{\alpha}_e \cap \den{\beta}_e \cup \big(\den{\beta}_e \setminus \den{\alpha}_c\big)$ and the result in the enunciate follows from $\den{\alpha\vee \beta}_e = \gamma_1 \cup \gamma_2$.
\end{proof}

\noindent In other words, equilibrium models of $\alpha \vee \beta$ consists of three possibilities: (1) common equilibrium models of $\alpha$ and $\beta$; (2) equilibrium models of $\alpha$ that are not classical models of $\beta$; and (3), vice versa, equilibrium models of $\beta$ that are not classical models of $\alpha$. Note that $\den{\alpha}_e \cap \den{\beta}_e \subseteq \den{\alpha \vee \beta}_e \subseteq \den{\alpha}_e \cup \den{\beta}_e$. As an example, consider the disjunction $p \vee (\neg p \rightarrow q)$ with $\alpha=p$ and $\beta=(\neg p \rightarrow q)$. The equilibrium models of each disjunct are $\den{p}_e=\{\str{20}\}$ and $\den{\neg p \rightarrow q}_e = \{\str{02}\}$, respectively. Obviously, $\alpha$ and $\beta$ have no common equilibrium model. Interpretation $\str{02}$ is an equilibrium model of $\beta$ and is not classical model of $\alpha$, and thus, it is an equilibrium model of $\alpha \vee \beta$. However, $\str{20}$ is both an equilibrium model of $\alpha$ and a classical model of $\beta$, and so it is disregarded. As a result, $\den{p \vee (\neg p \rightarrow q)}_e=\{\str{02}\}$.

As another example, take $r \vee (\neg p \rightarrow q)$. In this case, $\den{r}_e=\{\str{002}\}$ and $\den{\neg p \rightarrow q}_e=\{\str{020}\}$. Since each equilibrium model of one disjunct is not a classical model of the other disjunct, $\den{r \vee (\neg p \rightarrow q)}_e=\den{r}_e \cup \den{\neg p \rightarrow q}_e = \{\str{002},\str{020}\}$.

\section{Entailment relations}\label{sec:entail}

Logical entailment is usually defined by saying that the models of a formula (or a theory) are a subset of models of another formula (the entailed consequence). In our setting, we may consider different sets of models of a same formula $\alpha$: $\den\alpha$, $\den{\alpha}_c$ and $\den{\alpha}_e$. Therefore, it is not so strange that we can find different types of entailments for ASP in the literature. We summarize some of them in the following definition.

\begin{definition}
Given two formulas $\alpha, \beta$ we say that:

\begin{tabular}{rlcl}
$\alpha$ \emph{entails} $\beta$ (in $G_3$), & written $\alpha \models \beta$, & iff & $\den{\alpha} \subseteq \den{\beta}$\\
$\alpha$ \emph{classically entails} $\beta$, & written $\alpha \models_c \beta$, & iff & $\den{\alpha}_c \subseteq \den{\beta}_c$\\
$\alpha$ \emph{skeptically entails} $\beta$, & written $\alpha \models_{sk} \beta$, & iff & $\den{\alpha}_e \subseteq \den{\beta}_c$\\
$\alpha$ \emph{credulously entails} $\beta$, & written $\alpha \models_{cr} \beta$, & iff & $\den{\alpha}_e \cap \den{\beta}_c \neq \emptyset$\\
$\alpha$ \emph{weakly entails} $\beta$, & written $\alpha \models_{e} \beta$ & iff & $\den{\alpha}_{e} \subseteq \den{\beta}_{e}$\\
$\alpha$ \emph{strongly entails} $\beta$, & written $\alpha \models_{s} \beta$, & iff & for any formula $\gamma$, \\
& & & $\alpha \wedge \gamma \models_{e} \beta \wedge \gamma$, that is,\\
& & & $\den{\alpha \wedge \gamma}_e \subseteq \den{\beta \wedge \gamma}_e$ \qed
\end{tabular}

\end{definition}

The first two relations, $\models$ and $\models_c$, correspond to logical entailments in the monotonic logics of $G_3$ and classical propositional calculus, respectively. Obviously, $G_3$ entailment implies classical entailment (remember that $S_c=S\cap \cI_c$). The next two entailments, $\models_{sk}$ and $\models_{cr}$ are typically used for non-monotonic queries where $\alpha$ is assumed to be a program and $\beta$ some query in classical logic. In this way, $\beta$ is a skeptical (resp. credulous) consequence of $\alpha$ if any (resp. some) equilibrium model of $\alpha$ is a classical model of $\beta$. In~\cite{Pea06}, an  \emph{equilibrium entailment}, $\alpha \eqmodels \beta$, is defined as $\alpha \models_{sk} \beta$ when $\den\alpha \neq \cI$ and $\den{\alpha}_e \neq \emptyset$, and $\alpha \models_c \beta$ otherwise.

The direct entailment between two programs would correspond to $\models_e$ which we have called here \emph{weak entailment}. The idea is that $\alpha \models_e \beta$ means that the equilibrium models of program $\alpha$ are also equilibrium models of $\beta$. An operational reading of this entailment is that, in order to obtain equilibrium models for $\beta$, we can try solving $\alpha$ and, if a solution for the latter is found, it will also be a solution to the original program. If this same relation holds \emph{for any context} $\gamma$, i.e., we can replace $\beta$ by $\alpha$ inside some larger program and the solutions of the result are still solutions for the original program, then we talk about \emph{strong entailment}.

To the best of our knowledge, the strong entailment relation has not been studied in the literature although its induced equivalence relation, \emph{strong equivalence}~\cite{LPV01}, is well-known and was, in fact, one of the main motivations that originated the interest in ASP for $G_3$ and equilibrium logic. It is obvious that strong entailment implies weak entailment (it suffices with taking $\gamma=\top$). Using the previous entailment relations, we can define several equivalence relations by considering entailment in both directions. As a result, we get the following derived characterisations:

\begin{definition}
Given two formulas $\alpha, \beta$ we say that:

\begin{tabular}{rl@{\!}c@{\!}l}
$\alpha$ is \emph{equivalent} to $\beta$ (in $G_3$), & written $\alpha \equiv \beta$, & iff & $\den{\alpha}=\den{\beta}$\\
$\alpha$ is \emph{classically equivalent} to $\beta$, & written $\alpha \equiv_c \beta$, & iff & $\den{\alpha}_c = \den{\beta}_c$\\
$\alpha$ is \emph{weakly equivalent} to $\beta$, & written $\alpha \equiv_{e} \beta$ & iff & $\den{\alpha}_{e} = \den{\beta}_{e}$\\
$\alpha$ is \emph{strongly equivalent} to $\beta$, & written $\alpha \equiv_{s} \beta$, & iff & for any formula $\gamma$, \\
& & & $\den{\alpha \wedge \gamma}_e = \den{\beta \wedge \gamma}_e$. \qed
\end{tabular}
\end{definition}

Note how $\alpha \equiv_s \beta$ iff both $\alpha \models_s \beta$ and $\beta \models_s \alpha$. The following result is a rephrasing of the main theorem in~\cite{LPV01}.

\begin{theorem}[From Theorem~1 in~\cite{LPV01}]\label{th:se}
Two formulas $\alpha, \beta$ are strongly equivalent iff they are equivalent in $G_3$ (or HT). In other words: $\alpha \equiv_s \beta$ iff $\alpha \equiv \beta$.\qed
\end{theorem}

It is, therefore, natural to wonder whether this relation also holds for entailment, that is, whether strong entailment $\alpha \models_s \beta$ also corresponds to entailment\footnote{As a matter of fact, other authors~\cite{DSTW08,SL14} have implicitly or explicitly used HT entailment (i.e. our $G_3$ relation $\models$) as one of the two directions of strong equivalence without considering that there could exist a difference between $\models$ and $\models_s$ as we defined here.} in $G_3$, $\alpha \models \beta$.  However, it is easy to see that these two relations \emph{are different}. As a counterexample, let $\alpha=(p \vee q)$ and $\beta=(\neg p \rightarrow q)$ from Example~\ref{ex:1}. We can easily check that $\alpha \models \beta$: indeed, $\den{\alpha}=\{\str{20},\str{02},\str{21},\str{12},\str{22}\}$ $\subseteq \den{\beta}$ as we saw in Example~\ref{ex:1}. However, the interpretation $\str{20}$ ($p$ true and $q$ false) is an equilibrium model of $\alpha$ which is not equilibrium model of $\beta$. Thus, $\alpha \not\models_e \beta$ and so $\alpha \not\models_s \beta$ either, since weak entailment is obviously a necessary condition for strong entailment.

Fortunately, strong entailment can be compactly captured using the denotational semantics, as we prove next. We begin proving an auxiliary result. 
\begin{lemma}
\label{one_aux_lemma} \zlabel{one_aux_lemma}
Given any $v \in \cI$, let $\gamma_v$ be the formula: $$\gamma_{v} \eqdef \bigwedge_{v(p)=2} p$$ 
Then, for any formula $\alpha$ and any $v \in \den{\alpha}_c$, we have $v \in \den{\alpha \wedge \gamma_v}_e$.\qed
\end{lemma}

\begin{theorem}\label{th:sentail}
$\alpha \models_s \beta$ iff the following two conditions hold:
\begin{enumerate}
\item[(i)] $\alpha \models_c \beta$
\item[(ii)] $ \den{\alpha}_c \down \cap \, \den{\beta} \subseteq  \den{\alpha}$ 
\end{enumerate}
\end{theorem}
\begin{proof}
We are going to start proving that the two conditions are sufficient for $\alpha \models_{s} \beta$. Let us take any formula $\gamma$ and any $v \in \den{\alpha \wedge \gamma}_e$. Then both $v \in \den{\gamma}_c$ and $v \in \den{\alpha}_c \subseteq \den{\beta}_c$. Thus, $v \in \den{\beta \wedge \gamma}_c$. Suppose we had some $u < v$, $u \in \den{\beta \wedge \gamma}$. Then, $u \in \den{\alpha}_c \down$ because $u < v \in \den{\alpha}_c$. But then, $u \in \den{\alpha}_c \down \cap \, \den{\beta} \subseteq  \den{\alpha}$, and as $u \in \den{\gamma}$ too, we would get that $v$ is not in equilibrium for $\alpha \wedge \gamma$, reaching a contradiction.

For proving that the two conditions are necessary, suppose $\alpha \models_s \beta$. For (i), take $v \in \den{\alpha}_c$. Since $v \in \den{\alpha \wedge \gamma_v}_e$ because of Lemma \ref{one_aux_lemma}, it follows that $v \in \den{\beta \wedge \gamma_v}_e \subseteq \den{\beta}_c$. 

For (ii), take some $u \in \den{\alpha}_c \down \cap \, \den\beta$ and assume $u \not \in \den{\alpha}$. Since $u \in\den{\alpha}_c \down$ and $u \not\in\den\alpha$, we conclude $u_t \in \den{\alpha}_c$ and $u<u_t$. Consider the formula:
$$\gamma:= \gamma_{u} \wedge \bigwedge_{u(p)=u(q)=1} p \rightarrow q$$
\noindent so that, obviously, $u \in \den\gamma$. We are going to show that $u_t \in \den{\alpha \wedge \gamma}_e$ but $u_t \not \in \den{\beta \wedge \gamma}_e$ something that contradicts strong entailment. We begin observing that $u_t \not \in \den{\beta \wedge \gamma}_e$ because $u \in \den\gamma \cap \den\beta = \den{\gamma \wedge \beta}$ but $u<u_t$ so $u_t$ is not in equilibrium. Now, to show $u_t \in \den{\alpha \wedge \gamma}_e$, it is easy to see that $u_t \in \den{\alpha \wedge \gamma}_c$, since we had $u_t \in \den{\alpha}_c$ and $u \in \den\gamma$ implies $u_t\in \den{\gamma}_c$. To see that $u_t$ is in equilibrium, take any  $w \in \den{\alpha \wedge \gamma}$ such that $w < u_t$. Now, notice that $w_t=u_t$, but $w \in \den\gamma \subseteq \den{\gamma_u}$ and, thus, the only possibility is that $w \geq u$. Moreover $w > u$ because $w\in\den\alpha$ while $u \not\in \den\alpha$. From $u<w<u_t(=w_t)$ we conclude that there exists some atom $p$, $u(p)=1$ and $w(p)=2$, and some atom $q$, $w(q)=1$ and $w_t(q)=u_t(q)=2$. But then, $w(q)=1$ implies $u(q)=1$ too and we get $u(p)=u(q)=1$ so that implication $p \rightarrow q$ occurs in the conjunction in $\gamma$. However, $w(p)=2$ and $w(q)=1$ means that $w$ is not a model of $p\rightarrow q$, which contradicts the assumption $w \in\den{\alpha \wedge \gamma}$.
\end{proof}

The proof to show that (ii) is a necessary condition for strong entailment relies on showing that, if it does not hold, we can build a formula $\gamma$ (a logic program) for which $\alpha \wedge \gamma \not\models_e \beta \wedge \gamma$. In fact, this part of the proof is not new: it reproduces the logic program built in the proof for Theorem 1 in~\cite{LPV01} for strong equivalence. However, \cite{LPV01} did not explicitly consider the concept of strong entailment, nor its characterisation in terms of sets of models, as provided here in Theorem~\ref{th:sentail}. 

Once Theorem~\ref{th:sentail} is separated as an independent result, we can easily provide an immediate proof of Theorem~\ref{th:se}. Combining both entailment directions of $\alpha \equiv_s \beta$ amounts now to satisfying the three conditions: 
\begin{enumerate}
\item[(i)] $\den{\alpha}_c=\den{\beta}_c$
\item[(ii)] $\den{\alpha}_c \down \cap \, \den{\beta} \subseteq  \den{\alpha}$ 
\item[(iii)] $\den{\beta}_c \down \cap \, \den{\alpha} \subseteq  \den{\beta}$
\end{enumerate}
\noindent but as $\den\beta \subseteq \den{\beta}_c \down = \den{\alpha}_c \down$ and $\den\alpha \subseteq \den{\alpha}_c \down = \den{\beta}_c \down$ we eventually get: (i) $\den{\alpha}_c=\den{\beta}_c$; (ii) $\den{\beta} \subseteq  \den{\alpha}$; and (iii) $\den{\alpha} \subseteq  \den{\beta}$. But these, altogether, are equivalent to $\den\alpha = \den\beta$.

To conclude this section, we consider an application of Theorem~\ref{th:sentail}, providing a sufficient condition for strong entailment that may be useful in some cases. Suppose that, apart from condition (i) of Theorem~\ref{th:sentail}, we further had $\beta \models \alpha$. Then, condition (ii) would become trivial since $\den{\alpha}_c \down \cap \, \den{\beta} \subseteq  \den{\beta}$ and $\den{\beta} \subseteq \den\alpha$. Therefore:

\begin{corollary}\label{cor:ex}
If $\alpha \models_c \beta$ and $\beta \models \alpha$ then $\alpha \models_s \beta$.\qed
\end{corollary}

As an example, suppose we have a program $\Pi=\beta \wedge \gamma$ containing the disjunction $\beta=p \vee q$, typically used, for instance, to generate a choice between $p$ and $q$ in ASP. This formula is classically equivalent to $\alpha=(\neg p \rightarrow q) \wedge (\neg q \rightarrow p)$ which is also a common way for generating choices in ASP that does not use disjunction. Unfortunately, it is well-known that, in the general case $\alpha$ and $\beta$ are not strongly equivalent. For instance, if $\Pi=\beta \wedge (p\rightarrow q) \wedge (q \rightarrow p)$ we get the equilibrium model $\str{22}$ ($p$ and $q$ true) whereas for $\Pi'=\alpha \wedge (p\rightarrow q) \wedge (q \rightarrow p)$ we get no equilibrium model. However, $\beta \models \alpha$ in $G_3$ and, by Corollary~\ref{cor:ex}, if we replace $\beta$ by $\alpha$ in $\Pi$, any equilibrium (or stable) model we obtain in the new program will also be an equilibrium model of the original one (although, perhaps, we may lose equilibrium models with the replacement). Moreover, we can also replace $\beta=p \vee q$ by $\alpha'=\neg p \rightarrow q$ or by $\alpha''=\neg q \rightarrow p$ and the same property will still hold.

\section{Conclusions}\label{sec:conc}

We have introduced an alternative formulation of equilibrium models and its monotonic basis, Here-and-There (or, more precisely, G\"odel's three-valued logic $G_3$) that assigns a set of models (called a denotation) to each formula. This semantics, the main contribution of the paper, allows describing $G_3$, classical and equilibrium models using several compact set operations. Using denotations, we have proved again some already known fundamental results for $G_3$ or Equilibrium Logic to show that much textual effort usually done in the literature can be rephrased in terms of formal equivalences on sets of interpretations that, in many cases, even amount to simple properties from standard set theory. On the other hand, as side contributions or applications of this semantics, we have also obtained some additional fundamental results. For instance, we have proved that, while disjunction in $G_3$ is definable in terms of the other connectives, conjunction is a basic operation and cannot be derived from disjunction and implication. We have also shown that the equilibrium models of a disjunction can be obtained in a compositional way, in terms of the equilibrium and classical models of the disjuncts. Finally, we have defined (and characterised in denotational terms) a new type of entailment we called \emph{strong entailment}: a formula strongly entails another formula if the latter can be replaced by the former in any context while keeping a subset of the original equilibrium models.

A recent outcome of our current work is~\cite{CMMSE15} focused on the formulation of the denotational semantics using a theorem prover so that most of the meta-theorems for Equilibrium Logic and $G_3$ in this paper have been automatically checked using the PVS theorem prover~\cite{PVS92}. Future work includes the reformulation in denotational terms of different classes of models that are known to charaterise syntactic subclasses of logic programming~\cite{FINK11,FINK13} and the extension to the infinitary and first order versions of Equilibrium Logic. Finally, it would also be interesting to explore how the new definition of strong entailment can be applied in belief update or even inductive learning for ASP.

\newpage
\bibliographystyle{acmtrans}
\bibliography{refs}

\end{document}